\documentclass[11pt]{article} 

\usepackage[utf8]{inputenc} 

\usepackage{geometry} 
\usepackage[hidelinks]{hyperref}
\usepackage{graphicx} 
\usepackage{amsfonts,amsthm,amsmath,amssymb,esint}
\usepackage{slashed}
\usepackage{authblk}
\usepackage{tikz}
\usetikzlibrary{arrows}
\usepackage{stmaryrd}

\newcommand{\C}{\mathbb{C}}
\newcommand{\Z}{\mathbb{Z}}

\newcommand{\N}{\mathbb{N}}
\newcommand{\G}{\mathbb{G}}

\newcommand{\Pb}{\mathbb{P}}

\newcommand{\Sb}{\mathbb{S}}

\newcommand{\K}{\mathbb{K}}
\newcommand{\M}{\mathbb{M}}
\newcommand{\ket}{\rangle}
\newcommand{\bra}{\langle}

\newcommand{\Bi}{\mathcal{B}}

\newcommand{\Hi}{\mathcal{H}}

\newcommand{\Ei}{\mathcal{E}}

\newcommand{\Oi}{\mathcal{O}}

\newcommand{\Ti}{\mathcal{T}}

\newcommand{\GH}{\mathfrak{H}}

\newcommand{\GT}{\mathfrak{T}}

\newcommand{\Tr}{\operatorname{Tr}}

\newcommand{\SU}{\operatorname{SU}}
\newcommand{\Un}{\operatorname{U}}

\newtheorem{thm}{Theorem}

\newtheorem{Lemma}[thm]{Lemma}
\newtheorem{prop}[thm]{Proposition}
\newtheorem{cor}[thm]{Corollary}

\theoremstyle{definition}
\newtheorem{dfn}[thm]{Definition}
\newtheorem{Notation}[thm]{Notation}
\newtheorem{Remark}[thm]{Remark}
\newtheorem{Example}[thm]{Example}

\newcommand{\bone}{\mathbf{1}}

\usepackage{booktabs} 
\usepackage{array} 
\usepackage{paralist} 
\usepackage{verbatim} 
\usepackage{subfig} 

\usepackage{fancyhdr} 
\evensidemargin=\oddsidemargin
\usepackage{sectsty}
\allsectionsfont{\sffamily\mdseries\upshape} 

\usepackage[nottoc,notlof,notlot]{tocbibind} 
\usepackage[titles,subfigure]{tocloft} 

\title{Dequantization via quantum channels}
\author{Andreas Andersson}
\affil{\small Email: fornjotnr@hotmail.com}
\affil{\footnotesize Max Planck Institute for Mathematics in the Sciences, Inselstrasse 22, D-04103 Leipzig, Germany
\\Wollongong University, School of Mathematics and Applied Statistics, 2522 Wollongong, Australia}
\affil{Mathematics Subject Classification 2010 Primary: 81P16; Secondary: 94A40, 81R60}
\affil{Keywords: Quantum channels, completely positive maps, large-$N$ limit, Stinespring representation, quantization, noncommutative geometry, compact quantum groups, subproduct systems, complementary channel, Kraus operators.}

\begin{document}
\maketitle
\abstract
For a unital completely positive map $\Phi$ (``quantum channel") governing the time propagation of a quantum system, the Stinespring representation gives an enlarged system evolving unitarily. We argue that the Stinespring representations of each power $\Phi^m$ of the single map together encode the structure of the original quantum channel and provides an interaction-dependent model for the bath. The same bath model gives a ``classical limit" at infinite time $m\to\infty$ in the form of a noncommutative ``manifold" determined by the channel. In this way a simplified analysis of the system can be performed by making the large-$m$ approximation. These constructions are based on a noncommutative generalization of Berezin quantization. The latter is shown to involve very fundamental aspects of quantum-information theory, which are thereby put in a completely new light. 

\section*{Introduction}
One ``classical limit" of a quantum system is achieved (or rather defined) by taking the $\hbar\to 0$ limit. The term ``classical" is motivated by the behavior of the Feynman path integral and by the Heisenberg commutation relations. Another option is provided by the fact that a ``classical" phase space is always embedded in the quantum one via the coherent states. A ``classical limit" can also be defined by taking the limit $m\to\infty$ of some quantum number $m$, or as the high-temperature limit. Indeed, a physical meaning to $\hbar\to 0$ can only be given if we isolate the physical parameters of the system which appear together with $\hbar$ in all expressions. Both $\hbar\to 0$ and coherent states are well known in the mathematical literature, although there the focus has been almost exclusively on the so-called ``quantization" process starting from a classical phase space. While quantization can give some indirect insights, it is of little direct physical relevance. On the other hand, taking the classical limit (to be referred here to as \textbf{dequantization}) often leads to simplifications which makes it easier to deduce the behavior of a quantum system. Nevertheless, quantization in the mathematical sense provides a much more rigid and well-defined structure leaving less room for arbitrariness.

It is therefore attractive to have a rather unified framework for the classical limit which is in the spirit of mathematical treatments of quantization. The approach of Berezin \cite{Bere2}, \cite{Bere3}, \cite{Schl1} gives the classical limit both in terms of coherent states and large quantum numbers. However, there have been no guidelines for how to apply it systematically. In this paper we outline how to do that, but we shall need a recent generalization of Berezin quantization \cite{An6} in order to cover realistic physics.

The Heisenberg-picture evolution of a quantum system is in almost all applications describable by a quantum channel \cite{BP}, \cite{Lind1}, \cite{AlLe1}, i.e. a unital completely positive map $\Phi:\Bi(\Hi_0)\to\Bi(\Hi_0)$ on the algebra $\Bi(\Hi_0)$ of all bounded operators on the system Hilbert space $\Hi_0$. Equivalently, the Schrödinger-picture evolution is governed by a trace-preserving completely positive map $\Phi_*:\Bi(\Hi_0)_*\to\Bi(\Hi_0)_*$, again called a quantum channel, on the space $\Bi(\Hi_0)_*$ of trace-class operators.  One may regard the application $\rho\to\Phi_*(\rho)$ for a density matrix $\rho$ as the transformation taking $\rho$ from time $t=0$ to some later time $t=\tau$. The condition on complete positivity (and even positivity) may fail if we try to divide the interval $[0,\tau]$, i.e. there may be no maps $\Phi_{\tau-t}$, $\Phi_t$ such that $\Phi_{\tau-t}\circ\Phi_t=\Phi$ \cite{WoCi1}. Nevertheless, the discrete-time semigroup $(\Phi^m)_{m\in\N_0}$ consisting of the powers $\Phi^m=\Phi\circ\cdots\Phi$ of the single map $\Phi$ gives a Markovian description of the evolution. The map $\Phi$ takes a special role in the dynamical semigroup and is determined as corresponding to ``smallest" time duration for which the evolution is given by a quantum channel (or by the time discretization given by experimental data points).

The aim of this paper is to associate a classical limit to the evolution $\Phi^m$, obtained in the infinite-time limit $m\to\infty$. This limit coincides with the large-$N$ limit of the bath, where $N$ is the dimension of the bath Hilbert space. The structure of the limit is determined by the Kraus operators representing $\Phi$ as in equation \eqref{Krausrep} below, although the limit does not depend on the choice of such a representation. In particular, if these Kraus operators commute then the classical limit is encoded in a compact manifold (the ``dequantization manifold"). If they fail to commute then one obtains a ``manifold" $\M$ which has to be understood in the sense of noncommutative geometry \cite{Co}. In any case, there is a lot of geometry in the classical limit, and it will be studied in detail elsewhere (e.g. \cite{An6}). 

For finite time $m\in\N$, the observable algebra $\Bi(\GH_m)$ of the ``Stinespring bath model" proposed in this paper is a ``fuzzy version" \cite{MuSa1}, \cite{BKV1} of the dequantization manifold $\M$. Also these objects carry a lot of geometric information and are in some cases known in the mathematical literature. For most noncommutative manifolds $\M$, the geometry of their fuzzifications has to be developed, and this is work in progress. 

\section{The main point}\label{mainpointsec}
Let $\Hi_0$ be a Hilbert space with a countable basis. We consider a quantum channel $\Phi:\Bi(\Hi_0)\to\Bi(\Hi_0)$, i.e. a completely positive map such that $\Phi(\bone)=\bone$, where $\bone$ is the identity operator. Any such $\Phi$ has a \textbf{Kraus representation}, i.e. there are bounded operators $K_1,\dots,K_n$ for some $n\in\N\cup\{\infty\}$ such that \cite{Choi1}
\begin{equation}\label{Krausrep}
\Phi(A)=\sum^n_{k=1}K_k^*AK_k,\qquad \forall A\in\Bi(\Hi_0),
\end{equation}
and the ``unitality" $\Phi(\bone)=\bone$ says that
\begin{equation}\label{unitalityKraus}
\sum^n_{k=1}K_k^*K_k=\bone.
\end{equation}
We shall assume $n<\infty$ throughout, and that $n$ is the minimal integer for which there exists an $n$-tuple of operators $K_1,\dots,K_n$ such that \eqref{Krausrep} holds. A corollary to Stinespring's theorem says that there is a Hilbert space $\GH$ of dimension $n$ and a unitary operator $W\in\Bi(\Hi_0\otimes\GH)$ such that
\begin{equation}\label{unitarymodel}
\Phi(A)=\bra e_1|W^{-1}(A\otimes\bone)We_1\ket,
\end{equation}
where $e_1$ is a unit vector in $\GH$ (see e.g. \cite{CGLM}). Thus, $W$ models the one-step evolution $A\to\Phi(A)$ unitarily on a larger Hilbert space. Suppose next that we want a unitary model for the evolution $A\to\Phi^2(A):=\Phi\circ\Phi(A)$ up to time $t=2$. It will not be possible to find a representation \eqref{unitarymodel} for $\Phi^2$ using $\GH$. Rather, the Stinespring space of the quantum channel $\Phi^2$ is isomorphic to the Hilbert space $\GH_2$ obtained from the span of sums of products $K_j^*K_k^*$ of (the adjoints of) two Kraus operators for the original map $\Phi$. We have
\begin{equation}\label{inclusion2}
\GH_2\subseteq\GH^{\otimes 2},
\end{equation}
and it is the relations among the Kraus operators $K_j$ which determines how big the subspace \eqref{inclusion2} really is. More generally, the Stinespring space of the quantum channel $\Phi^m$ is identifiable with the span $\GH_m$ of products $K_{j_m}^*\cdots K_{j_1}^*$ and (see the proof of \cite[Thm. 1.12]{Ar3})
\begin{equation}\label{inclusionintensor}
\GH_m\subseteq\GH^{\otimes m},\qquad \forall m\in\N_0.
\end{equation} 
\begin{equation}\label{inclusionm}
\GH_{m+l}\subseteq\GH_m\otimes\GH_l,\qquad \forall m,l\in\N_0.
\end{equation} 
Here $\N_0:=\{0,1,2,\dots\}$ and $\GH_0:=\C$. The whole semigroup $(\Phi^m)_{m\in\N_0}$ can then be obtained by tracing over a unitary dynamics $\bigoplus_m W_m$ on $\Hi_0\otimes\GH_\N$, where $\GH_\N$ is the Hilbert space
\begin{equation}\label{PhiFock}
\GH_\N:=\bigoplus_{m\in\N_0}\GH_m,
\end{equation} 
which we refer to as the \textbf{$\Phi$-Fock space}. 
\begin{Example}[Bosons]\label{fullcommute}
If the Kraus operators satisfy $[K_j,K_k]=0$ for all $j,k=1,\dots,n$, but no other relations, then 
$$
\GH_2=\GH^{\vee 2}:=\GH^{\otimes 2}\ominus\{e_j\otimes e_k-e_k\otimes e_j\}_{1\leq j<k\leq n},
$$
where $e_1,\dots,e_n$ denotes an orthonormal basis for $\GH$. In that case, $\GH_m$ is the subspace of totally symmetric vectors, and $\GH_\N=\GH^{\vee\N}$ as in \eqref{PhiFock} is the symmetric (``Bosonic") Fock space over $\GH$. This kind of maps $\Phi$ arises from tracing over the apparatus in a quantum measurement of $n$ commuting observables \cite{BP}. In particular, such $\Phi$'s appear as the reduced evolution of the position state in a (homogeneous) ``quantum walk" in the sense of \cite{ADZ}, \cite{VeAn1}. For space-inhomogeneous quantum walks, noncommuting Kraus operators are needed. 
\end{Example} 
\begin{Example}[White noise]\label{Fullfockexample}
 Kraus operators realized as large random matrices cannot be commuted through each other in any way, and there are no relations whatsoever between them. In that case we have equality $\GH_m=\GH^{\otimes m}$ in \eqref{inclusionintensor} and $\GH_\N=\GH^{\otimes \N}$ is the full (``free" or ``Boltzmannian") Fock space over $\GH$. 
\end{Example} 
To make the discussion about the structure provided by the Stinespring representations of all the $\Phi^m$'s more automatic, we observe that it fits into the following context \cite{ShSo1}, \cite{DRS1}. 
\begin{dfn}\label{Defsubprod}
A collection $\GH_\bullet=(\GH_m)_{m\in\N_0}$ of finite-dimensional Hilbert spaces with $\GH_0=\C$ and for which \eqref{inclusionm} holds is called a \textbf{subproduct system}. A subproduct system $\GH_\bullet$ is \textbf{commutative} if $\GH_m\subseteq\GH^{\vee m}$ for all $m$, where $\vee$ is the symmetric tensor product.
\end{dfn}
\begin{dfn}\label{Stinesubprod}
Let $\Phi$ be a quantum channel on $\Bi(\Hi_0)$. The collection $\GH_\bullet=(\GH_m)_{m\in\N_0}$ of Hilbert spaces such that $\Hi_m:=\Hi_0\otimes\GH_m$ is the minimal Stinespring representation of $\Phi^m$ for each $m$ will be referred to as the \textbf{Stinespring subproduct system} of $\Phi$.
\end{dfn}
An obvious but crucial fact is the following.
\begin{prop}
The Stinespring subproduct system of a quantum channel $\Phi$ is independent of the choice of Kraus representation of $\Phi$.
\end{prop}
\begin{Remark}[Homogeneous relations] The Stinespring subproduct system $\GH_\bullet$ depends on and only on \emph{homogeneous} relations among the Kraus operators. For instance, a relation such as $K_1K_2=K_3$ will not have any effect on $\GH_\bullet$.
\end{Remark}

The bath model \eqref{PhiFock} is the minimal one needed to obtain a unitary dilation of the entire semigroup $(\Phi^m)_{m\in\N_0}$. It should not be interpreted as describing a bath which exists as a system in its own right. Rather, it contains the interaction degrees of freedom which have an influence on the dynamics on $\Bi(\Hi_0)$ and all quantities associated with it describes the interaction only. As we have seen in Example \ref{fullcommute}, if $\GH_\bullet=\GH^{\vee\bullet}$ then the system interacts with a bath of independent harmonic oscillators. 
\begin{Remark}[Other approaches] 
The so-called ``repeated-interaction model" \cite{AtPa1} was introduced to give a unitary model for the whole semigroup $(\Phi^m)_{m\in\N_0}$, mainly inspired by a very specific experimental setup. It became clear however that it could be useful more generally if the ``bath" is interpreted as an interaction-dependent bath \cite{BBP1}, in the spirit of how we propose to interpret the Stinespring bath. The Stinespring bath has the advantage that it contains so much information about the channel $\Phi$, as we shall see.
\end{Remark}

\section{Inductive limits}
We have obtained a model $\Bi(\GH_\bullet)=(\Bi(\GH_m))_{m\in\N_0}$ for the bath which depends on the quantum channel $\Phi$. It turns out \cite[Thm. 4.8]{An6} that the subproduct structure of $\GH_\bullet$ ensures the existence of unital completely positive maps 
\begin{equation}\label{indlimsyst}
\iota_{m,l}:\Bi(\GH_m)\to\Bi(\GH_{m+l}),\qquad m\leq l
\end{equation}
satisfying $\iota_{r,l}\circ\iota_{m,r}=\iota_{m,l}$ for all $m\leq r\leq l$ and an ``asymptotic multiplicativity" condition. The maps $\iota_{m,l}$ allows us to translate the bath in time. In particular, there is a canonical way of constructing an ``infinite-time limit" ($m\to\infty$) of such a system $(\Bi(\GH_\bullet),\iota_{\bullet,\bullet})$ \cite{BlKi1}, \cite[§B]{Hawk1}. This limit is not the usual algebraic inductive limit where the $\iota_{m,l}$'s are required to be multiplicative, i.e. to satisfy $\iota_{m,l}(AB)=\iota_{m,l}(A)\iota_{m,l}(B)$ for all $A,B\in\Bi(\GH_m)$. Multiplicativity holds iff $\GH_\bullet=\GH^{\otimes\bullet}$, which is not what we want in general.

In order to define the infinite-time limit $C^*$-algebra $\Bi_\infty$, let us introduce the left and right shift operators $S_k$ and $R_j$ on Fock space $\GH_\N$ for $j,k=1,\dots,n$ as
$$
S_k(\psi):=p_{m+1}(e_k\otimes\psi),\qquad R_j(\psi):=p_{m+1}(\psi\otimes e_j)
$$
for all $\psi\in\GH_m\subset\GH_\N$ and all $m\in\N_0$, where $p_m:\GH^{\otimes m}\to\GH_m$ denotes the projection and $e_1,\dots,e_n$ is the orthonormal basis for $\GH$ corresponding to $K_1^*,\dots,K_n^*$. For words $\mathbf{k}=k_1\cdots k_m$ in $\{1,\dots,n\}$ we shall use the notation
$$
S_\mathbf{k}:=S_{k_1}\cdots S_{k_m},\qquad S_\mathbf{k}^*:=(S_\mathbf{k})^*=S_{k_m}^*\cdots S_{k_1}^*,
$$
and similarly for $R_\mathbf{k}$ and $K_\mathbf{k}$. The maps \eqref{indlimsyst} can now be defined by
$$
\iota_{m,l}(A):=\sum_{|\mathbf{r}|=l-m}R_\mathbf{r}AR_\mathbf{r}^*\big|_{\GH_l},\qquad \forall A\in\Bi(\GH_m),
$$
where the sum is over all multi-indices of length $l-m$. 
\begin{dfn}
The \textbf{Toeplitz algebra} $\Ti_\GH$ is the $C^*$-algebra generated by the left shifts $S_1,\dots,S_n$, and the \textbf{Toeplitz core} $\Ti_\GH^{(0)}$ is the $C^*$-subalgebra of $\Ti_\GH$ generated by $S_jS_k^*$ for all $j,k=1,\dots,n$. 
\end{dfn}

The Toeplitz core $\Ti_\GH^{(0)}$ in fact coincides with the algebra of all sequences $A_\bullet=(A_m)_{m\in\N_0}$ with $A_m\in\Bi(\GH_m)$ such that 
$$
\iota_{m,l}(A_m)=A_l,\qquad \forall l\geq m
$$
for some sufficiently large $m$ \cite[Lemma 4.10]{An6}. Such a sequence $A_\bullet$ may be called ``eventually constant" (under $\iota_{\bullet,\bullet}$). Let $\Gamma_0\subset\Ti_\GH^{(0)}$ be the ideal of all sequences which are eventually $0$. Then
\begin{equation}\label{indlimdef}
\Bi_\infty:=\Ti_\GH^{(0)}/\Gamma_0
\end{equation}
is the desired infinite-time limit $C^*$-algebra, referred to as the (generalized) \textbf{inductive limit} of the sequence $(\Bi(\GH_\bullet),\iota_{\bullet,\bullet})$.

\begin{Remark}\label{bimodrem}
Similarly one constructs inductive limits $\Ei_k$ of the sequences $\Bi(\GH_\bullet)\otimes\GH_k$ for all $k\in\Z$, where $\GH_{-m}:=\GH_m^*$ for $m>0$. Each $\Ei_k$ is a bimodule over $\Bi_\infty$, and it contains $\GH_k$ as a subspace. 
 Take $m\in\N$. Letting the left multiplication operator on $\Ei_m$ defined by $f\in\Bi_\infty$ be denoted simply by $f$, we define the level-$m$ \textbf{Toeplitz operator} with symbol $f$ to be the operator $\breve{\varsigma}^{(m)}(f)$ on $\GH_m$ given by
\begin{equation}\label{Toepop}
\breve{\varsigma}^{(m)}(f)\psi:=\Pi_m(f\psi),\qquad\forall \psi\in\GH_m,
\end{equation}
where $\Pi_m:\Ei_m\to\GH_m$ is the projection (which is necessary because the action of $f$ does not preserve $\GH_m$ unless $f\propto \bone$).
\end{Remark}

\begin{Remark}\label{remarkcovsymb}
The adjoint $\varsigma^{(m)}:\Bi(\GH_m)\to\Bi_\infty$ of the Toeplitz map $\breve{\varsigma}^{(m)}$ with respect to suitable inner product on $\Bi_\infty$ and $\Bi(\GH_m)$ simply takes a normally ordered element $S_jS_k^*$ of $\Ti_\GH^{(0)}$ to its image in the quotient \eqref{indlimdef}. We refer to $\varsigma^{(m)}(A)$ as the \textbf{covariant symbol} of an operator $A\in\Bi(\GH_m)$; cf. Remark \ref{projvarrem} below. 
\end{Remark}

In the following, when we say the ``algebraic part" of $\Bi_\infty$ we mean the $*$-algebra $^{0}\Bi_\infty$ generated by the images of the shifts $S_jS_k^*$ in $\Ti_\GH^{(0)}/\Gamma_0$ without taking norm closure. 

\begin{thm}[{\cite[Cor. 5.26]{An6}}]\label{thmstrictquant}
The sequence $(\Bi(\GH_m),\breve{\varsigma}^{(m)})_{m\in\N_0}$ is a strict quantization of $\Bi_\infty$ in the sense that for all $f,g\in {^{0}\Bi_\infty}$ we have (cf. \cite[Def. 1.1.1]{Lan1})
\begin{enumerate}[(i)]
\item{$\lim_{m\to\infty}\|\breve{\varsigma}^{(m)}(f)\|=\|f\|$ \textnormal{(Rieffel's condition)},}
\item{$\lim_{m\to\infty}\|\breve{\varsigma}^{(m)}(fg)-\breve{\varsigma}^{(m)}(f)\breve{\varsigma}^{(m)}(g)\|=0$ \textnormal{(von Neumann's condition)},}
\item{$\lim_{m\to\infty}\|m^{-1}[\breve{\varsigma}^{(m)}(f),\breve{\varsigma}^{(m)}(g)]-\{f,g\}\|=0$ \textnormal{(Dirac's condition)},}
\end{enumerate} 
and for each $m\in\N$, every operator in $\Bi(\GH_m)$ is of the form $\breve{\varsigma}^{(m)}(f)$ for some $f\in{^{0}\Bi_\infty}$. Here $\{\cdot,\cdot\}$ is the Poisson bracket defined in op. cit.). 
\end{thm}
\begin{Example}[Projective varieties]\label{projvarrem}
If $\GH_\bullet$ is commutative (recall that this means that $\GH_m\subseteq\GH^{\vee m}$ for all $m$) then $\Bi_\infty\cong C(M)$ is the $C^*$-algebra of continuous functions on a (complex nonsingular) projective variety $M\subseteq\C\Pb^{n-1}$, and $\GH_m$ is the Hilbert space of holomorphic sections of the $m$th power of a certain (``pre-quantum") line bundle $L$ over $M$,
\begin{equation}\label{Kahlerhilbs}
\GH_m=H^0(M,L^{\otimes m})
\end{equation}
The bimodule $\Ei_m$ mentioned in Remark \ref{bimodrem} is the space of continuous sections of the line bundle $L^{\otimes m}$. The function $\varsigma^{(m)}(A)\in C^\infty(M)$ (cf. Remark \ref{remarkcovsymb}) is the (unique) ``Berezin covariant symbol" of an operator $A$ on $\GH_m$, and if $A=\breve{\varsigma}^{(m)}(f)$ for some function $f\in C(M)$ (every operator on $\GH_m$ is of this form, as a special case of Theorem \ref{thmstrictquant}) then $f$ is a (non-unique) \textbf{contravariant symbol} of $\breve{\varsigma}^{(m)}(f)$ \cite{Bere3}, \cite{Schl1}, \cite{CGR1}. The fact that $(\Bi(\GH_m),\breve{\varsigma}^{(m)})_{m\in\N_0}$ is a strict quantization of $C^\infty(M)$ has been known for a long time \cite[Thms. 4.1, Thm. 4.2, Prop. 4.2, §5]{BMS}. 

We remark that coadjoint orbits of compact Lie groups are projective varieties. All compact Kähler manifolds which are ``quantizable" are projective varieties \cite[§2.1]{BeSl1}, and conversely every projective variety is a quantizable compact Kähler manifold. They are complex submanifolds of projective space. See \cite[§3]{An6} for a more detailed summary of the commutative case.
\end{Example}

\section{The dequantization manifold}
We have seen in Example \ref{projvarrem} that, for a projective variety $M$ and with $\GH$ defined by \eqref{Kahlerhilbs}, the algebras $\Bi(\GH_m)$ give an increasingly better approximation of $C^\infty(M)$ the larger the value of $m$. As a result, a quantum channel $\Phi$ with commuting Kraus operators determines a classical manifold. Similarly, a general channel $\Phi$ (i.e. with not necessarily commutating Kraus operators) comes with algebras $\Bi(\GH_m)$, and we have associated an algebra $\Bi_\infty$ in the same way as in the commutative case (recall the inductive limit defined in \eqref{indlimdef}). The results of the last section demands the following definition. 
\begin{dfn} We define
\begin{equation}\label{defofdequanteq}
C(\M):=\Bi_\infty
\end{equation}
to be the ``algebra of continuous functions" on the \textbf{dequantization manifold} $\M$ associated to the quantum channel $\Phi$. We also write $C^\infty(\M):={^{0}\Bi_\infty}$, where we recall that ${^{0}\Bi_\infty}$ is a dense $*$-subalgebra of $\Bi_\infty$. 
\end{dfn}
\begin{Notation}
To be perfectly clear: there is no honest manifold unless $\Phi$ defines a commutative subproduct system. From now on we use the symbol $\M$ for general $\Phi$ (including the commutative case) but clarify $M:=\M$ when we restrict to the commutative ones. 
\end{Notation}

\begin{Remark}[Why the inductive limit is classical]
The inductive limit \eqref{defofdequanteq} ignores all the finite-time fluctuation observables. Only the operators which have an effect that survives in the infinite-time limit $m\to\infty$ give an element of $C(\M)$. 
\end{Remark}
The relevance of the dequantization manifold $\M$ can be understood as follows. The algebraic relations among the Kraus operators are present already in the quantum system $\Hi_0$; all information is there. However, we need infinite time (infinitely many repetitions of the transformation $\Phi$) until all this information becomes relevant for every property of the system. At time $m=1$ the system already depends on the Kraus operators. However, the algebraic Kraus relations are irrelevant at this stage. As $m$ increases, the \textbf{trajectories}  
\begin{equation}\label{trajectories}
K_{j_m}^*\cdots K_{j_1}^*AK_{j_1}\cdots K_{j_m}
\end{equation}
of an observable $A\in\Bi(\Hi_0)$ depend more and more on the ``quantum symmetry" of the $K_j$'s (i.e. their algebraic relations). At infinite time this dependence is so strong that the ``manifold of trajectories" $\M$ is expected to cover most, or perhaps all, of the information of the system (i.e. of the infinite-time limit).

 Fix a quantum channel $\Phi$. Denote by $\GT$ the $C^*$-algebra generated by the Kraus operators $K_1,\dots,K_n$ of $\Phi$. There is a $\Z$-grading on $\GT$, which we write as
\begin{equation}\label{Krausalgedec}
\GT=\overline{\bigoplus_{k\in\Z}\GT^{(k)}}^{\|\cdot\|},
\end{equation}  
where $\GT^{(k)}$ is the norm-closed span of sums of products of the form $K_\mathbf{r} K_\mathbf{s}^*$ and $K_\mathbf{s}^*K_\mathbf{r}$ with $|\mathbf{r}|-|\mathbf{s}|=k$.
By \eqref{trajectories}, the physically relevant expressions $K_{j_m}^*\cdots K_{j_1}^*K_{k_1}\cdots K_{k_l}$ are the trajectories
\begin{equation}\label{idtrajectories}
K_{j_m}^*\cdots K_{j_1}^*K_{k_1}\cdots K_{k_m}
\end{equation}
(i.e. the ones with $m=l$; in the next section we will see why we need all the $K_j^*K_k$'s and not only the $K_j^*K_j$'s). The zeroth component $\GT^{(0)}$ is a $C^*$-algebra and contains the elements \eqref{idtrajectories} but also the elements
\begin{equation}\label{adjointtrajectories}
K_{r_1}\cdots K_{r_m}K_{s_m}^*\cdots K_{s_1}^*.
\end{equation}
\begin{Lemma}[Inductive limit versus algebra of Kraus operators]\label{normordrem}
Suppose that
\begin{enumerate}[(i)]
\item{the relations among the $K_j$'s are homogeneous, and}
\item{``normal ordering" is possible in $\GT^{(0)}$, i.e. there is some  way of transforming an expression of the form \eqref{adjointtrajectories} into the form \eqref{idtrajectories}.}
\end{enumerate}
Then 
\begin{equation}\label{trajectorymanequaldequant}
\GT^{(0)}\cong C(\M).
\end{equation}
\end{Lemma}
\begin{proof} The assumption that $n$ is the minimal number of Kraus operators needed to represent $\Phi$ as in \eqref{Krausrep} implies that $K_1,\dots,K_n$ are linearly independent. As observed, the products $K_\mathbf{j}^*$ with $|\mathbf{j}|=m$ span a Hilbert space $\GH_m$. We can identify $K_\mathbf{j}^*K_\mathbf{k}$ for $|\mathbf{j}|=m=|\mathbf{k}|$ with an overcomplete basis for $\Bi(\GH_m)$, and the same is true for $Z_\mathbf{j}Z_\mathbf{k}^*$. Indeed, assumption (i) implies that the generators $Z_1,\dots,Z_n$ of $\Oi_\GH$ satisfy the same non-$*$ relations as the generators $K_1,\dots,K_n$ of $\GT$. Since normal ordering is assumed to be possible in $\GT^{(0)}$, the $\Bi(\GH_m)$'s form a dense $*$-subalgebra of $\GT^{(0)}$. It follows that the $K_\mathbf{j}^*K_\mathbf{k}$'s satisfy the same relations as the $Z_\mathbf{j}Z_\mathbf{k}^*$'s and generate isomorphic $*$-algebras. Finally, the norm on $\Bi_\infty$ (see \cite[§4.1]{An6}) is seen to correspond to the operator norm on $\Bi(\Hi_0)$ under the identification of these $*$-algebras. So we have an isomorphism $\GT^{(0)}\cong\Bi_\infty\equiv C(\M)$. 


\end{proof}
Intuitively, if a normal-ordering prescription is available in $\GT^{(0)}$ then it is encoded in $\GH_\bullet$ because (in that case) the anti-normally ordered elements $K_\mathbf{j}K_\mathbf{k}^*$ belong to $\Bi(\GH_m)$ as well. In general, the inductive limit differs from $\GT^{(0)}$. The issue is that $C(\M)$ contains only the information of the \emph{forward} time evolution $(\Phi^m)_{m\in\N_0}$. A suitable notion of time reversal for $\Phi$ is a channel $\tilde{\Phi}$ whose trajectories are combinations of those in \eqref{adjointtrajectories} (see \cite{An4}). The anti-normally ordered trajectories \eqref{adjointtrajectories}, which give the \emph{backward} evolution, may not retrodictable from the forward ones, and that is when \eqref{trajectorymanequaldequant} fails. Except for extreme irregularity in the time evolution however, we expect that normal ordering is possible in $\GT^{(0)}$. 

In contrast to Example \ref{projvarrem}, the infinite-time limit of the algebras $\Bi(\GH_m)$ is not commutative in general. However, we shall see in Example \ref{Examplefreegroup} that, even in the opposite extreme of free commutation relations, there are still reasons to regard $C^\infty(\M)$ as a classical limit of the bath algebras.

\section{The complementary channel}
We consider again a fixed quantum channel $\Phi:\Bi(\Hi_0)\to\Bi(\Hi_0)$ with a finite number of Kraus operators $K_1,\dots,K_n$ and let $\Phi_*:\Bi(\Hi_0)_*\to\Bi(\Hi_0)_*$ be the corresponding evolution of density matrices, so
$$
\Phi_*(\rho)=\Tr_{\GH}\big(W(\rho\otimes|e_1\ket\bra e_1|)W^{-1}\big),\qquad \forall \rho\in\Bi(\Hi_0)_*
$$
where $\GH\cong\C^n$ and $W$ are as in \eqref{unitarymodel}. Recall that for finite-dimensional $\Hi_0$, the \textbf{complementary channel} of $\Phi_*$ is the quantum channel $\Phi^\natural_*:\Bi(\Hi_0)_*\to\Bi(\GH)_*$ defined by \cite{DFH1}
\begin{equation}\label{complemchannel}
\Phi^\natural_*(\rho):=\Tr_{\Hi_0}\big(W(\rho\otimes|e_1\ket\bra e_1|)W^{-1}\big),\qquad \forall \rho\in\Bi(\Hi_0)_*,
\end{equation}
which yields the expression
$$
\Phi^\natural_*(\rho)=\sum^n_{j,k=1}\Tr(\rho K_k^*K_j )S_jS_k^*\big|_{\GH}.
$$
In \cite[§1]{KMNR1}, the state $\Phi^\natural_*(\rho)$ on $\GH$ is interpreted as the information available about the state $\rho$ on $\Hi_0$ at time $t=\tau$ (if $\Phi$ evolves the system from $t=0$ to $t=\tau$). 

We now want to find a suitable analogue of \eqref{complemchannel} in the Heisenberg picture for which this ``information-at-time-$\tau$" picture remains valid. This can be done by interpreting the collection $\{K_j^*K_k\}_{j,k=1}^n$ as describing a measurement. Let $\rho_0$ be a fixed density matrix on $\Hi_0$ (the ``true state" of the system), assumed to satisfy $K_k\rho_0K_k^*\ne 0$ for all $k=1,\dots,n$. Here we allow $\Hi_0$ to be infinite-dimensional again. For a projective measurement, i.e. one with $K_jK_k=\delta_{j,k}K_k$ and $K_j^*=K_j$, the information known about $A\in\Bi(\Hi_0)$ at time $t=\tau\equiv 1$ is given by the expectation values
$$
\Psi^{(1)}(A)_{k,k}:=\frac{\Tr(\rho_0 K_k^*K_kA)}{\Tr(\rho_0 K_k^*K_k)},\qquad k=1,\dots,n
$$
of $A$ in the ``post-measurement states"\footnote{Usually the post-measurement states are taken to be proportional to $K_k\rho_0 K_k^*$, but it will be more convenient to \\use $\rho_0K_k^*K_k$ so that we have perfect agreement with the conventions in \cite{An6}.} $\rho_j:=\rho_0K_k^*K_k/\Tr(K_k\rho_0K_k^*)$. For general $K_j$'s, we need also 
$$
\Psi^{(1)}(A)_{j,k}:=\frac{\Tr(\rho_0 K_k^*K_jA)}{\Tr(\rho_0K_k^* K_k)},\qquad j\ne k,
$$
because there is a probability of misinterpreting outcome $k\in\{1,\dots,n\}$ as being $j\in\{1,\dots,n\}$. The same considerations apply for all times $t=m\in\N$. The ``dequantization" of an operator $A$ on $\Hi_0$ should therefore be determined by the post-measurement expectations $\Psi^{(1)}(A)_{j,k}$ and their analogues for larger $m$. For sufficiently many repetitions $m$, these expectation values should then approximate the operator $A$ in a $\Phi$-dependent fashion. 

We want to apply results from \cite{An6}, so we need to identify our variables with those used there.
\begin{Notation} For each $m\in\N$, we let $Q_m/\Tr(Q_m)\in\Bi(\GH_m)$ be the \textbf{correlation matrix} of the state $\rho_0$ on $\Hi_0$, i.e. we define the matrix $Q_m$ with entries $Q_{\mathbf{j},\mathbf{k}}:=\bra e_\mathbf{j}|Q_me_\mathbf{k}\ket$ by 
\begin{equation}\label{corrmatrix}
\frac{Q_{\mathbf{j},\mathbf{k}}}{\Tr(Q_m)}:=\Tr(\rho_0K_\mathbf{k}^*K_\mathbf{j}),
\end{equation}
where we fix $\Tr(Q_m)$ e.g. by requiring that $\Tr(Q_m^{-1})=\Tr(Q_m)$. 
\end{Notation}
It is \eqref{unitalityKraus} which ensures that the correlation matrix is a density matrix. Note that $Q_m$ is invertible even if $\rho_0$ is not faithful, because of the assumption that $K_k\rho_0K_k^*\ne 0$ for all $k=1,\dots,n$. 

In order to apply the whole machinery of \cite{An6} we need $\rho_0$ to have $\Phi$-\textbf{symmetric correlations}, in the sense that if $Q:=Q_1$ then we have the two equalities
$$
Q_m=p_mQ^{\otimes m}p_m=Q^{\otimes m}p_m
$$
for all $m\in\N$, where $p_m:\GH^{\otimes m}\to\GH_m$ denotes the projection. Such a $\rho_0$ may be regarded as the asymptotic equilibrium state of the system (see \cite{An4}). With this assumption on $\rho_0$, we can change the inner product on each $\GH_m$ so that it depends on $Q_m$ without spoiling the subproduct condition \eqref{inclusionm} (see \cite[§4.5]{An6}).

We are thus led to the following definition.
\begin{dfn} Suppose that $\rho_0$ is a density matrix on $\Bi(\Hi_0)$ with $\Phi$-symmetric correlations. For $m\in\N$, the \textbf{time-$m$ dequantization} of $A\in\Bi(\Hi_0)$ is the operator on $\GH_m$ given by
\begin{equation}
\Psi^{(m)}(A)=\Tr(Q_m)\sum_{|\mathbf{j}|=m=|\mathbf{k}|}(Q^{\otimes m})^{-1}_{\mathbf{k},\mathbf{k}}\Tr(\rho_0K_\mathbf{k}^*K_\mathbf{j}A) S_\mathbf{j}S_\mathbf{k}^*\big|_{\GH_m}\label{dequantascontr}.
\end{equation}
\end{dfn}
For $m=1$ we indeed get
$$
\Psi^{(1)}(A)=\sum_{j,k=1}^n\frac{\Tr(\rho_0K_k^*K_jA)}{\Tr(\rho_0K_k^*K_k)} S_jS_k^*\big|_{\GH}
$$
which is what we had in our discussion about post-measurement expectation values in the beginning of the section. For general $m\geq 1$ however, the normalization in $\Psi^{(m)}(A)$ is given by products $\Tr(\rho_0K_{k_1}^*K_{k_1})\cdots \Tr(\rho_0K_{k_m}^*K_{k_m})$ of probabilities for the time $m=1$ measurement outcomes. We have no intuitive argument for why that normalization would be the correct thing, except that it ensures that $\Psi^{(m)}(\bone)=p_m$ ($=$ the unit in $\Bi(\GH_m)$) while the naive choice $\Tr(\rho_0K_\mathbf{k}^*K_\mathbf{k})$ would not. 
 The reason why we believe \eqref{dequantascontr} is the correct thing is the following. 
\begin{thm}\label{dequantthm}
Let $\Phi:\Bi(\Hi_0)\to\Bi(\Hi_0)$ be a quantum channel and let $\GH_\bullet$ be its subproduct system. Suppose that $\rho_0$ is a density matrix with $\Phi$-symmetric correlations such that $\Tr(\rho_0\cdot)$ restricts to a faithful state on the $C^*$-algebra $\GT^{(0)}$ generated by $K_j^*K_k$ and $K_kK_j^*$ ($j,k=1,\dots,n$). Let $\Bi^\infty$ be the space of operators on Fock space $\GH_\N$ spanned by the elements
$$
\Psi(A):=(\Psi^{(m)}(A))_{m\in\N_0},\qquad A\in\Bi(\Hi_0),
$$  
where $\Psi^{(m)}(A)$ is defined in \eqref{dequantascontr}. Then $\Bi^\infty$ is a $C^*$-algebra under the multiplication
\begin{equation}\label{projlimmult}
\Psi(A)\cdot\Psi(B):=\lim_{m\to\infty}\Psi^{(m)}(AB),\qquad \forall A,B\in\Bi(\Hi_0);
\end{equation}
in fact
$$
\Bi^\infty\cong C(\M),
$$
where $\M$ is the dequantization manifold of $\Phi$. 
\end{thm}
\begin{proof} Each $\Psi^{(m)}(A)$ is surjective. The definition of $\Psi^{(m)}$ ensures compatibility with the projective system discussed in \cite{An6} (where $Z_k$ corresponds to $K_k^*$) and so the result follows from \cite[Thm. 5.16]{An6}.
\end{proof}
 
\begin{cor}\label{ifhomogencor} Suppose that $K_1,\dots,K_n$ satisfy only homogeneous relations and allow normal ordering, so that we can identify $C(\M)=\GT^{(0)}\subset\Bi(\Hi_0)$ (see Lemma \ref{normordrem}). Suppose also that the restriction of $\Tr(\rho_0\cdot)$ to $\GT^{(0)}$ is faithful. Then the map $\Psi:\Bi(\Hi_0)\to\Bi^\infty$ restricts to an isomorphism $\breve{\varsigma}:C(\M)\to\Bi^\infty$. 
\end{cor}

\begin{Remark}[Multiplication] The multiplication in $\Bi^\infty$ coincides with that in $\Bi(\Hi_0)$ only for elements in $C(\M)$. For general $A\in\Bi(\Hi_0)$, the operator $\Psi(A)$ acts in a $\Phi$-dependent way.
\end{Remark}
The correlation matrix $Q/\Tr(Q)$ is a quantum analogue of the probability measure on $\{1,\dots,n\}$ which determines a classical measurement. The idea that a quantum measurement $\{K_k^*K_j\}_{j,k=1}^n$ gives a coarse-grained description of a possibly infinite-dimensional system $\Bi(\Hi_0)$ as an $n$-dimensional matrix algebra $\Bi(\GH)$ is not new; see \cite[§10.1]{AlFa1}. Here we have observed that we can increase the accuracy of such a description using possibly infinitely many repetitions of the same measurement.

The adjoint $\varsigma^{(m)}:\Bi(\GH_m)\to C(\M)$ of $\breve{\varsigma}^{(m)}:C(\M)\to\Bi(\GH_m)$ is a generalization of the covariant symbol map of classical Berezin quantization (cf. Remark \ref{projvarrem}). In fact, $\varsigma^{(m)}$ has appeared in the literature already, as the ``coarse-graining map" of the measurement \cite[§10.1]{AlFa1}; it takes the form
\begin{equation}\label{covsymbKraus}
\varsigma^{(m)}(X)=\sum_{|\mathbf{j}|=m=|\mathbf{k}|}X_{\mathbf{j},\mathbf{k}}K_\mathbf{j}^*K_\mathbf{k}
\end{equation}
for $X\in\Bi(\GH_m)$, where $X_{\mathbf{j},\mathbf{k}}\in\C$ are the matrix entries of $X$. Finally, a straightforward calculation shows that if $V_m:\Hi_0\to\Hi_0\otimes\GH$ is the Stinespring isometry of $\Phi^m$ (see e.g. \cite{Ar3}) then
$$
\varsigma^{(m)}(X)=V_m^*(\bone\otimes X)V_m,
$$
so that $\varsigma^{(m)}$ in fact coincides with what has also been referred to as the Heisenberg-picture complementary channel of $\Phi^m$ in some works \cite[§IV]{KSW1}. 
\begin{Remark}[Limit state] If $\rho_0$ has $\Phi$-symmetric correlations then the states $\Tr(Q_m\cdot)/\Tr(Q_m)$ on the $\Bi(\GH_m)$'s converge to the state on $C(\M)\subset\Bi(\Hi_0)$ which is the restriction of $\Tr(\rho_0\cdot)$. This result follows from \cite[Remark 2.7]{An4}, \cite[Prop. 5.9]{An6} and makes precise in what way the correlations matrices approximate the state $\rho_0$. 
\end{Remark}

\section{Cutting off at finite time}
At finite times $m$, the observable algebra $\Bi(\GH_m)$ of the bath $\GH_m$ is still noncommutative even for commuting Kraus operators. Only in the infinite-time approximation do we get the algebra $C^\infty(\M)$. 

An exception is when $\Bi(\GH_{m+1})\cong\Bi(\GH_m)$ for all $m$ larger than some integer $l$. Then we simply have $C^\infty(\M)=\Bi(\GH_l)$. In this case the classical limit is reached at finite time $l$.
\begin{Example}[Projective von-Neumann measurements]\label{onetimefuzzyex}
Let $P_1,\dots,P_n$ be pairwise orthogonal projections in $\Bi(\Hi_0)$ with $\sum_kP_k=\bone$ and set
$$
\Phi(A):=\sum^n_{k=1}P_kAP_k,\qquad\forall A\in\Bi(\Hi_0).
$$
The subproduct system $\GH^{\bullet}=(\GH_m)_{m\in\N_0}$ defined by $\Phi$ has $\GH_2\subset\GH^{\otimes 2}$ equal to  the subspace defined by the relation $P_jP_k=\delta_{j,k}P_k$ for all $j,k=1,\dots,n$. That is, 
$$
\GH_2=\text{span}\{e_j\otimes e_j\}_{j=1}^n\cong\text{span}\{e_j\}_{j=1}^n=\GH,
$$
and more generally $\GH_m\cong\GH$ for all $m$. By commutativity we know that the dequantization manifold $\M=M$ is a projective variety (see Example \ref{projvarrem}). It is defined by the homogeneous ideal in $\C[z_1,\dots,z_n]$ generated by the polynomials
$$
z_jz_k,\qquad j\ne k=1,\dots,n.
$$
The points $z=(z_1,\dots,z_n)$ in $\C^n$ that satisfy $z_jz_k=0$ for all $j\ne k$ are of the form $(\lambda,0,\dots,0)$, or $(0,\lambda,0,\dots,0)$, etc., up to $(0,\dots,0,\lambda)$ for some $\lambda\in\C$. These points form $n$ lines in $\C^n$ which give $n$ distinct points in in $\C\Pb^{n-1}$. The dequantization manifold $\M=M$ is therefore an $n$-point space $\{1,\dots,n\}$. The algebra $C(M)$ is in this peculiar case isomorphic to the algebra $\GT$ generated by the Kraus operators $P_1,\dots,P_n$ (see \eqref{Krausalgedec}). It is acted upon ergodically by the permutation group $\Pb_n$ of $n$ letters. In this example, cutting off at finite time doesn't change anything; we obtain classicality already after one time step $m=1$. 
\end{Example}

\begin{Example}[Noncommuting projective measurements] The simplest possible example of noncommuting Kraus operators is provided by a channel $\Phi:\Bi(\Hi_0)\to\Bi(\Hi_0)$ which the sum of two noncommuting projections $P$ and $Q$ on $\Hi_0$,
$$
\Phi(A):=\frac{1}{2}(PAP+QAQ).
$$
The map $\Phi$ corresponds to two sequential projective von Neumann measurements with only two possible outcomes each. Even though $P^2=P$ and $Q^2=Q$, it may happen if $\Hi_0$ is infinite-dimensional that we do not get $\GH_m\cong\GH_l$ for any $l\ne m$. That is, the possible trajectories $P, Q, PQP, (\bone-P)QP,\dots$ of any number $m\in\N$ of measurements  are infinitely many and hence the $\Phi$-Fock space is infinite-dimensional and contains a lot of information. Such maps were proposed as efficients tools in quantum state tomography in \cite{NG}
\end{Example}

In most applications to open quantum system, e.g. in time-resolved spectroscopy, we would rather want a gradual time development of the system response that can be monitored over time, not a collapse immediately from the initial state with no data points for the time in between. In particular, if we want to have any information about coherence (the off-diagonal matrix elements of the density matrix in the chosen basis) then we should not wipe out the off-diagonals of the density matrix in one step. 

Therefore, the projections $P_k$ of Example \ref{onetimefuzzyex} need to be replaced by general operators $K_k$ with $\sum_kK_k^*K_k=\bone$. In the limit of infinite time we again get classicality if $[K_j,K_k]=0$ for all $j,k$ but in the more realistic finite-time situation, the cutoff at time $m\in\N$ leads to a fuzzy space if the $K_k$'s are sufficiently far from being projections. For noncommuting $K_k$'s and finite time we get a ``fuzzy noncommutative manifold" of trajectories.

\section{Examples of dequantization manifolds}\label{Exampmansection}
A discussion about the physical meaning of different Kraus-operator commutation relations can be found in \cite{An4}. Here we simply record what dequantization manifolds they give rise to. 
\begin{Example}[$2$-sphere]\label{sphereex}
For $\Phi:\Bi(\Hi_0)\to\Bi(\Hi_0)$ with a minimal Kraus decomposition consisting of two commuting Kraus operators satisfying no other relations, $\GH_\bullet=\GH^{\vee\bullet}$ is the subproduct system in Example \ref{fullcommute} with $n=2$, for which $\GH_m=\GH^{\vee m}$ is the spin-$m/2$ representation of $\SU(2)$. In particular, $\dim(\GH_m)=m+1$. The dequantization manifold is here $M=\Sb^2$, the classical $2$-sphere. The finite-time algebra $\Bi(\GH_m)$ is the ``fuzzy $2$-sphere" at deformation parameter $m\in\N$ \cite{Mad1}. 
\end{Example}

\begin{Example}[Projective spaces]\label{Exampleprojspaces}
More generally, with $n$ Kraus operators satisfying commutativity only, $\GH_m$ is an irreducible representation of $\SU(n)$ and the dequantization manifold 
$$
M=\C\Pb^{n-1}
$$
is the projective $n$-space. The isomorphism $\C\Pb^1\cong\Sb^2$ recovers Example \ref{sphereex}. Just as for $n=2$, the fuzzy projective spaces $\Bi(\GH_m)$ for higher $n$ have been widely studied \cite{BDLMC}, \cite{KuSa1}. 
\end{Example}
\begin{Example}[$q$-deformed projective spaces] Taking one more step towards higher generality, $n$ Kraus operators satisfying the $q$-commutation relations of the first column in the defining representation of the quantum group $\SU_q(n)$ gives a dequantization manifold
$$
\M=\C\Pb^{n-1}_q
$$
which is known as the $q$-deformed projective $n$-space \cite{KhMo1}. Again, the corresponding fuzzy manifolds $\Bi(\GH_m)$ are not completely unknown.
\end{Example}
\begin{Example}[Coadjoint orbits]\label{Examplecoadjoint}
For Kraus operators satisfying the relations of the coordinate functions on a compact Lie group $G$, the dequantization manifold is a coadjoint orbit $M\cong G/K$. We recover Example \ref{Exampleprojspaces} if we take $G=\SU(n)$. Then $K=\Un(1)\times\SU(n-1)$ and $G/K$ is complex projective $n$-space $\C\Pb^{n-1}$.
\end{Example}
\begin{Example}[Quantum homogeneous spaces]\label{Exampleqspaces}
Building on the last example, suppose that the Kraus operators satisfy the relations of the first column (or row) of the defining unitary representation $u\in\Un(n)\otimes C(\G)$ of a compact matrix quantum group $\G$ (see \cite{KlS}, \cite{Wor1} for the notion of quantum groups). Again, the dequantization manifold, which we by analogy of Example \ref{Examplecoadjoint} write as $\G/\K:=\M$, is a ``quantum homogeneous space", in the sense that $C(\G)$ coacts ergodically on $C(\G/\K)$. For these results we need to assume that normal ordering is possible in $C(\G/\K)$. Note that if $\Hi_0$ is finite-dimensional then $\G$ must be finite, i.e. $C(\G)$ is a finite-dimensional $C^*$-algebra.
\end{Example}
\begin{Example}[Half-commutative]\label{halflibex}
If $K_j^*=K_j$ and $K_jK_kK_l=K_lK_kK_j$ for all $j,k,l=1,\dots,n$ then $\M$ is in fact an ordinary manifold, namely $\C\Pb^{n-1}$ \cite{BaGo1}. 
\end{Example}
\begin{Example}[Free commutation relations]\label{Examplefreegroup}
Consider now Kraus operators satisfying free commutation relations, i.e. no relations at all. The algebra $C(\M)$ is the $\text{UHF}(n^\infty)$ algebra, the $\Un(1)$-invariant part of the Cuntz algebra $\Oi_n$. The algebras $\Bi(\GH_m)$ are (as always) just finite-dimensional matrix algebras whereas the generators of $\Oi_n$ satisfy free commutation relations. So in this case $m\to\infty$ behaves in the same way as the large-$N$ limits reviewed in \cite{GMW1}. 
\end{Example}

\section{Summary and outlook}
Let us record the main points of our discussion.
\begin{enumerate}[(i)]
\item{A discrete semigroup $(\Phi^m)_{m\in\N_0}$ of quantum channels is a very general description of the time evolution of a quantum system.}
\item{The collection $\GH_\bullet=(\GH_m)_{m\in\N_0}$ of finite-dimensional Hilbert spaces coming from the Stinespring representations of $(\Phi^m)_{m\in\N_0}$ has a convenient structure and depends only on the map $\Phi$, not on the choice of Kraus operators.}
\item{The bath model $\GH_\N=\bigoplus_{m\in\N_0}\GH_m$ depends on the particular dynamics and does not describe a quantum system (``bath") which exists on its own; it is the bath needed to account for the given dynamics $\Phi^m$.}
\item{The algebra $C^\infty(\M)$ of ``functions" on the dequantization manifold $\M$ is the infinite-time approximation of the bath observable algebras $\Bi(\GH_m)$ at finite times $m\in\N$. It can be interpreted as the algebra of functions on a manifold of infinite-time trajectories. }
\item{The information about an operator $A$ on $\Hi_0$ gathered up to time $m$ is contained in an operator $\Psi^{(m)}(A)$ on $\GH_m$, and the map $\Psi^{(m)}$ is the complementary channel of $\Phi$ (modified to depend on the relevant state $\rho_0$ on $\Hi_0$).} 
\item{Every operator $A$ on $\Hi_0$ dequantizes to the operator $\Psi(A)$ on Fock space $\GH_\N$ which acts on the subspace $\GH_m$ as $\Psi^{(m)}(A)$. With a new $\Phi$-dependent multiplication defined by taking the limit $m\to\infty$, the $\Psi(A)$'s form an algebra isomorphic to $C^\infty(\M)$.}\label{functfromobs}
\item{The maps $\breve{\varsigma}^{(m)}$ and $\varsigma^{(m)}$ which defines the noncommutative Berezin quantization of $C(\M)$ can be identified with well-known objects in quantum information theory.}
\item{The restriction of the state $\Tr(\rho_0\cdot)$ to $C(\M)$ is recovered as the infinite-time limit of the correlation matrices $Q_m/\Tr(Q_m)$.}
\item{For channels $\Phi$ with a certain symmetry under a compact classical or quantum group $\G$, the dequantization manifold is a quantum homogeneous space $\G/\K$ under this group.}\label{grouppoint}
\end{enumerate}
 
In the case of \eqref{grouppoint}, the channel $\Phi$ gives rise to a structure resembling that of algebraic superselection theory \cite{DHR1} and there is a kind of group duality not entirely different from \cite{DR1}, as will be discussed  in a separate paper. 

When the dequantization manifold is an honest manifold, it carries a lot of geometric structure. One may ask for something similar with noncommuting Kraus operators. That is indeed possible by means of noncommutative geometry, and future work will be devoted to develop further the geometry of dequantization manifolds and their fuzzy versions. The perhaps most intriguing point is that the eigenvectors of a suitable ``Laplacian" on $\M$ are the inequivalent $\Phi$-trajectories, and so the Gibbs partition function associated with this Laplacian coincides with the partition function defined as the integral over all paths (Feynman integral). 

Building on the physical interpretations of the Kraus operators discussed in \cite{An4}, one may try to find real systems with given dequantization manifolds. Doing so would further deepen our understanding of dissipative evolutions and check the validity of quantum channels as models for them.

\section*{References}

\bibitem[ADZ]{ADZ} Aharonov Y, Davidovich L, Zagury N. Quantum random walks. Physical Review A. Vol 48, Issue 2, p. 1687 (1993).

\bibitem[AlFa1]{AlFa1} Alicki R, Fannes M. Quantum dynamical systems. Oxford University Press, Oxford (2001).

\bibitem[AlLe1]{AlLe1} Alicki R, Lendi K. Quantum dynamical semigroups and applications. Lecture Notes in Phys. Vol 717, Springer-Verlag Berlin Heidelberg (2007).

\bibitem[An4]{An4} Andersson A. Detailed balance as a quantum-group symmetry of Kraus operators. arXiv: 1506.00411 (2015). 

\bibitem[An6]{An6} Andersson A. Berezin quantization of noncommutative projective varieties. arXiv: submit/1271363 (2015). 

\bibitem[Ar3]{Ar3} Arveson W. The index of a quantum dynamical semigroup. Journal of functional analysis. Vol 146, Issue 2, pp. 557-588 (1996).

\bibitem[AtPa1]{AtPa1} Attal S, Pautrat Y. From repeated to continuous quantum interactions. Annales Henri Poincaré. Vol 7, Issue 1, pp. 59-104. Birkhäuser-Verlag (2006).

\bibitem[BDLMC]{BDLMC} Balachandran AP, Dolan BP, Lee JH, Martin X, O’Connor D. Fuzzy complex projective spaces and their star-products. J. Geom. Phys. Vol 43, pp. 184 (2002).

\bibitem[BKV1]{BKV1} Balachandran AP, Kürkçüoglu S, Vaidya S. Lectures on fuzzy and fuzzy SUSY physics. Singapore: World Scientific (2007).

\bibitem[BaGo1]{BaGo1} Banica T,  Goswami D. Quantum isometries and noncommutative spheres. Comm. Math. Phys. Vol 298, Issue2, pp. 343-356 (2010). 

\bibitem[Bere2]{Bere2} Berezin FA. General concept of quantization. Commun. math. Phys. Vol 40, pp. 153-174 (1995).

\bibitem[Bere3]{Bere3} Berezin FA. Covariant and contravariant symbols of operators. Mathematics of the USSR-Izvestiya. Vol 6, Issue 5, p. 1117 (1972).

\bibitem[BlKi1]{BlKi1} Blackadar B, Kirchberg E. Generalized inductive limits of finite-dimensional $C^*$-algebras. Math. Ann. Vol 307, pp. 343-380 (1997).

\bibitem[BeSl1]{BeSl1} Berceanu S, Schlichenmaier M. Coherent state embeddings, polar divisors and Cauchy formulas. J.Geom. Phys. Vol, Issue 34, pp. 336-358 (2000).

\bibitem[BMS]{BMS} Bordemann M, Meinrenken E Schlichenmaier M. Toeplitz quantization of Kähler manifolds and $\operatorname{gl}(n),n\to\infty$ limits. Commun. Math. Phys. Vol 165, Issue 2, pp. 281–296 (1994).

\bibitem[BP]{BP} Breuer HP, Petruccione F. The theory of open quantum systems. Oxford University Press, Oxford (2003).

\bibitem[BBP1]{BBP1} Bruneau L, De Bievre S, Pillet CA. Scattering induced current in a tight-binding band. Journal of Mathematical Physics. Vol 52, Issue 2, p. 022109 (2011).

\bibitem[CGR1]{CGR1} Cahen M, Gutt S, Rawnsley J. Quantization of Kähler manifolds I: geometric interpretation of Berezin's quantization. JGP Vol 7, Issue 1 (1990). 

\bibitem[CGLM]{CGLM} Caruso F, Giovannetti V, Lupo C, Mancini S. Quantum channels and memory effects. Reviews of Modern Physics. Vol 86, Issue 4, p. 1203 (2014).

\bibitem[Choi1]{Choi1} Choi MD. Completely positive Linear maps on complex matrices. Lin. Alg. Appl. Vol 10, pp. 285–290 (1975).

\bibitem[Co]{Co} Connes A. Noncommutative geometry. Academic Press, Inc. (1994).

\bibitem[DFH1]{DFH1} Datta N, Fukuda M, Holevo AS. Complementarity and additivity for covariant channels. Quantum Information Processing. Vol 5, Issue 3, pp. 179-207 (2006).

\bibitem[DRS1]{DRS1} Davidson KR, Ramsey C, Shalit OM. The isomorphism problem for some universal operator algebras. Advances in Mathematics. Vol 228, Issue 1, pp. 167-218 (2011).

\bibitem[DHR1]{DHR1} Doplicher S, Haag R, Roberts JE. Fields, observables and gauge transformations I. Commun. Math. Phys. Vol 13, pp. 1-23 (1969).

\bibitem[DR1]{DR1} Doplicher S, Roberts JE. Endomorphisms of $C^*$-algebras, crossed products and duality for compact groups. Ann. Math. 130, 75–119 (1989).

\bibitem[GMW1]{GMW1} Guhr T, Müller-Groeling, A, Weidenmüller HA. Random-matrix theories in quantum physics: common concepts. Physics Reports. Vol 299, Issue 4, pp. 189-425 (1998).

\bibitem[Hawk1]{Hawk1} Hawkins E. Quantization of equivariant vector bundles. Commun. Math. Phys. Vol 202, Issue 3, pp. 517-546 (1999).

\bibitem[KhMo1]{KhMo1} Khalkhali M,  Moatadelro A. Noncommutative complex geometry of the quantum projective space. J. Geom. Phys. Vol 61, pp. 2436-2452 (2011).

\bibitem[KMNR1]{KMNR1} King C, Matsumoto K, Nathanson M, Ruskai MB. Properties of conjugate channels with applications to additivity and multiplicativity. arXiv: quant-ph/0509126 (2005).

\bibitem[KlS]{KlS} Klimyk AU, Schmüdgen K. Quantum groups and their representations. Vol. 552, Springer, Berlin (1997).

\bibitem[KSW1]{KSW1} Kretschmann D, Schlingemann D, Werner RF. The information-disturbance tradeoff and the continuity of Stinespring's representation. Information Theory, IEEE Transactions on. Vol 54, Issue 4, pp. 1708-1717 (2008).

\bibitem[KuSa1]{KuSa1} Kürkçüoglu S, Sämann C. Drinfeld twist and general relativity with fuzzy spaces. Classical and Quantum Gravity. Vol 24, Issue 2, p. 291 (2007).

\bibitem[Lan1]{Lan1} Landsman NP.  Mathematical topics between classical and quantum mechanics.  Springer (1998).

\bibitem[Lind1]{Lind1} Lindblad G. On the generators of quantum dynamical semigroups. Commun. Math. Phys. Vol 48, pp. 119-130 (1976).

\bibitem[Mad1]{Mad1} Madore J. An introduction to noncommutative differential geometry and its physical applications. Cambridge
University Press, Second edition (1999).

\bibitem[MuSa1]{MuSa1}  Murray S, Sämann C. Quantization of flag manifolds and their supersymmetric extensions. Adv. Theor.
Math. Phys. Vol 12, pp. 641-710 (2008).

\bibitem[NG]{NG} Navascués M, Pérez-García D. Sequential strong measurements and the heat vision effect. New J. Phys. Vol 13, 113038 pp. 1-18 (2011).

\bibitem[Schl1]{Schl1} Schlichenmaier M. Berezin-Toeplitz quantization for compact Kähler manifolds. A review of results. Adv. Math.l Phys. (2010).

\bibitem[ShSo1]{ShSo1} Shalit OM, Solel B. Subproduct systems. Documenta Mathematica. Vol 14, pp. 801-868 (2009).

\bibitem[VeAn1]{VeAn1} Venegas-Andraca SE. Quantum walks: a comprehensive review. Quantum Information Processing. Vol 11, Issue 5, pp. 1015-1106 (2012).

\bibitem[WoCi1]{WoCi1} Wolf MM, Cirac JI. Dividing quantum channels. Commun. Math. Phys. Vol 279, Issue 1, pp. 147-168 (2008).

\bibitem[Wor1]{Wor1} Woronowicz SL. Compact matrix pseudogroups. Commun. Math. Phys. Vol 111, pp. 613-665 (1987).

\end{document}